\documentclass[runningheads,a4paper]{llncs}

\usepackage{amssymb}
\usepackage{bbold} 
\usepackage{mathtools} 
\usepackage[vlined,linesnumberedhidden,ruled,resetcount]{algorithm2e} 
\usepackage{xspace} 
\usepackage{tikz}
\usepackage{a4wide}
\usetikzlibrary{automata}

\newcommand{\C}{\mathcal{C}}

\newcommand{\M}{\mathbb{M}}
\newcommand{\N}{\mathbb{N}}
\renewcommand{\O}{\mathcal{O}}

\newcommand{\Z}{\mathbb{Z}}

\newcommand{\zero}{\mathbb{0}}
\newcommand{\paths}{\text{Paths}}

\newcommand{\inp}{\text{in}}
\newcommand{\out}{\text{out}}

\newcommand{\varstates}{X_Q}
\newcommand{\varpaths}{X_P}
\newcommand{\varin}{X_I}
\newcommand{\varout}{X_O}
\newcommand{\trs}{\mathcal{A}_{out}}
\newcommand{\logic}{\textsf{PL}}
\newcommand{\Var}{\text{Var}}
\renewcommand{\dim}{\text{\#}}
\newcommand{\nfa}{\text{NFA}}
\newcommand{\trans}{\text{Trans}}
\newcommand{\suma}{\text{Sum}}
\newcommand{\init}{\ensuremath{\textsf{init}}}
\newcommand{\final}{\ensuremath{\textsf{final}}}
\newcommand{\conv}{\mathbin{\otimes}}
\newcommand{\bigconv}{\mathop{\bigotimes}}
\newcommand{\new}{\mathrel{\ensuremath{\stackrel{\makebox[0pt]{\mbox{\tiny def}}}{_=}}}}

\newcommand{\prefix}{\ensuremath{\mathrel{\sqsubseteq}}}
\newcommand{\sem}[1]{\ensuremath{[\![#1]\!]}}
\newcommand{\first}{\ensuremath{\triangleleft}}
\newcommand{\last}{\ensuremath{\triangleright}}

\usepackage{times}

\title{A Pattern Logic for Automata with Outputs\thanks{We warmly
    thank the anonymous reviewers for their helpful comments, and
    Isma\"el Jecker for spotting a bug in a preliminary version of the
    paper. E. Filiot is a research associate of F.R.S.-FNRS.
		He is supported by the French ANR Project ExStream (ANR-13-JS02-0010), the ARC Project
                Transform F\'ed\'eration Wallonie-Bruxelles and the
                FNRS CDR project J013116F. N. Mazzocchi is a PhD
                student funded by a FRIA fellowship from the
                F.R.S.-FNRS. J.-F. Raskin is supported by an ERC Starting Grant (279499: inVEST), by the ARC project
		$-$ Non-Zero Sum Game Graphs: Applications to Reactive Synthesis and Beyond $-$
		funded by the F\'ed\'eration Wallonie-Bruxelles, and by a Professeur Francqui de Recherche grant awarded by the Francqui Fondation.
}}
\author{
	Emmanuel Filiot,
	Nicolas Mazzocchi,
	and Jean-Fran\c{c}ois Raskin
}
\institute{Universit\'e libre de Bruxelles}

\begin{document}
	\maketitle
	
	\begin{abstract}
            We introduce a logic to express structural properties
            of automata with string inputs and, possibly,  outputs in
            some monoid. In this logic, the set of predicates talking about
            the output values is parametric, and we provide sufficient
            conditions on the predicates under which the
            model-checking problem is decidable. We then consider
            three particular automata models (finite automata,
            transducers and automata weighted by integers --
            sum-automata --) and instantiate the
            generic logic for each of them. We give tight complexity
            results for the three logics and the model-checking
            problem, depending on whether the formula is fixed or
            not. We study the expressiveness of our logics by
            expressing classical structural patterns characterising 
            for instance finite ambiguity and
            polynomial ambiguity in the case of finite automata,
            determinisability and finite-valuedness in the case of
            transducers and sum-automata. Consequently to our complexity results, we directly obtain
            that these classical properties can be decided in
            \textsc{PTime}.
	\end{abstract}

	\newcommand{\figureScale}{.8}

\newcommand{\drawPcp}{
	\scalebox{\figureScale}{
		\begin{tikzpicture}[>=stealth, node distance=3.5cm, thick]
			\node[state, accepting, initial left, initial text =] (q1) {};
			\node[state, accepting, initial left, initial text =, right of = q1] (q2) {};
		
			\path[->]
				(q1) edge [loop above] node[above] {$
					\begin{array}{c|c}
						i \in \{1, \dots, n\} & u_i
					\end{array}
					$} (q1)
				(q2) edge [loop above] node[above] {$
					\begin{array}{c|c}
						i \in \{1, \dots, n\} & v_i
					\end{array}
					$} (q2)
			;
		\end{tikzpicture}
	}
}

\newcommand{\drawBTP}{
	\scalebox{\figureScale}{
		\begin{tikzpicture}[>=stealth, node distance=2cm, thick]
		\node[state, initial, initial text=] (q00) {$q_{0, 0}$};
		\node[state, right of = q00] (q10) {$q_{1, 0}$};
		\node[state, right of = q10] (q20) {$q_{2, 0}$};
		\node[state, right of = q20] (qk0) {$q_{k, 0}$};
		
		\node[state, initial, initial text=, below of = q00] (q01) {$q_{0, 1}$};
		\node[state, right of = q01] (q11) {$q_{1, 1}$};
		\node[state, right of = q11] (q21) {$q_{2, 1}$};
		\node[state, right of = q21] (qk1) {$q_{k, 1}$};
		
		\node[state, initial, initial text=, below of = q01] (q0k) {$q_{0, k}$};
		\node[state, right of = q0k] (q1k) {$q_{1, k}$};
		\node[state, right of = q1k] (q2k) {$q_{2, k}$};
		\node[state, right of = q2k] (qkk) {$q_{k, k}$};
		
		\path[->]
		(q00) edge node[above] {$
			\begin{array}{c|c}
			u_{1, 0} & v_{1, 0}
			\end{array}
			$} (q10)
		(q10) edge [loop above] node[above] {$
			\begin{array}{c|c}
			u'_{1, 0} & v'_{1, 0}
			\end{array}
			$} (q10)
		(q10) edge node[above] {$
			\begin{array}{c|c}
			u_{2, 0} & v_{2, 0}
			\end{array}
			$} (q20)
		(q20) edge [loop above] node[above] {$
			\begin{array}{c|c}
			u'_{2, 0} & v'_{2, 0}
			\end{array}
			$} (q20)
		(q20) edge[dashed] (qk0)
		(qk0) edge [loop above] node[above] {$
			\begin{array}{c|c}
			u'_{k, 0} & v'_{k, 0}
			\end{array}
			$} (qk0)
		
		(q01) edge node[above] {$
			\begin{array}{c|c}
			u_{1, 1} & v_{1, 1}
			\end{array}
			$} (q11)
		(q11) edge [loop above] node[above] {$
			\begin{array}{c|c}
			u'_{1, 1} & v'_{1, 1}
			\end{array}
			$} (q11)
		(q11) edge node[above] {$
			\begin{array}{c|c}
			u_{2, 1} & v_{2, 1}
			\end{array}
			$} (q21)
		(q21) edge [loop above] node[above] {$
			\begin{array}{c|c}
			u'_{2, 1} & v'_{2, 1}
			\end{array}
			$} (q21)
		(q21) edge[dashed] (qk1)
		(qk1) edge [loop above] node[above] {$
			\begin{array}{c|c}
			u'_{k, 1} & v'_{k, 1}
			\end{array}
			$} (qk1)
		
		(q0k) edge node[above] {$
			\begin{array}{c|c}
			u_{1, k} & v_{1, k}
			\end{array}
			$} (q1k)
		(q1k) edge [loop above] node[above] {$
			\begin{array}{c|c}
			u'_{1, k} & v'_{1, k}
			\end{array}
			$} (q1k)
		(q1k) edge node[above] {$
			\begin{array}{c|c}
			u_{2, k} & v_{2, k}
			\end{array}
			$} (q2k)
		(q2k) edge [loop above] node[above] {$
			\begin{array}{c|c}
			u'_{2, k} & v'_{2, k}
			\end{array}
			$} (q2k)
		(q2k) edge[dashed] (qkk)
		(qkk) edge [loop above] node[above] {$
			\begin{array}{c|c}
			u'_{k, k} & v'_{k, k}
			\end{array}
			$} (qk0)
		;
		\end{tikzpicture}
	}
}

\newcommand{\drawWeakTwinning}{
	\scalebox{\figureScale}{
		\begin{tikzpicture}[>=stealth, node distance=2cm, thick]
			\node[state, initial, initial text=] (qi) {$q_0$};
			\node[state, right of = qi] (q1) {$q_1$};
			\node[state, right of = q1] (q2) {$q_2$};
			\node[state, accepting, right of = q2] (qf) {$q_3$};
			
			\path[->]
				(qi) edge (q1)
				(q1) edge [loop above] node[above] {$
					\begin{array}{c|c}
						u & u_1
					\end{array}
					$} (q1)
				(q1) edge [loop below] node[below] {$
					\begin{array}{c|c}
						v & v_1
					\end{array}
					$} (q1)
				(q1) edge node[above] {$
					\begin{array}{c|c}
						u & u_2
					\end{array}
					$} (q2)
				(q2) edge [loop above] node[above] {$
					\begin{array}{c|c}
						v & v_2
					\end{array}
					$} (q2)
				(q2) edge (qf)
			;
		\end{tikzpicture}
	}
}

\newcommand{\drawTwinning}{
	\scalebox{\figureScale}{
		\begin{tikzpicture}[>=stealth, node distance=2cm, thick]
		\node[state, initial, initial text=] (q1) {$q_1$};
		\node[state, right of = q1] (p1) {$p_1$};
		\node[state, accepting, right of = p1] (r1) {$r_1$};
		\node[state, initial, initial text=, below of = q1] (q2) {$q_2$};
		\node[state, right of = q2] (p2) {$p_2$};
		\node[state, accepting, right of = p2] (r2) {$r_2$};
		
		\path[->]
		(q1) edge node[above] {$
			\begin{array}{c|c}
			u & v_1
			\end{array}
			$} (p1)
		(p1) edge [loop above] node[above] {$
			\begin{array}{c|c}
			u' & v'_1
			\end{array}
			$} (p1)
		(p1) edge node[above] {$
			\begin{array}{c|c}
			u'' & v''_1
			\end{array}
			$} (r1)
		(q2) edge node[above] {$
			\begin{array}{c|c}
			u & v_2
			\end{array}
			$} (p2)
		(p2) edge [loop above] node[above] {$
			\begin{array}{c|c}
			u' & v'_2
			\end{array}
			$} (p2)
		(p2) edge node[above] {$
			\begin{array}{c|c}
			u'' & v''_2
			\end{array}
			$} (r2)
		;
		\end{tikzpicture}
	}
}

\newcommand{\drawFiniteAmbiguity}{
	\scalebox{\figureScale}{
		\begin{tikzpicture}[>=stealth, node distance=2cm, thick]
			\node[state, initial, initial text=] (qi) {$q_0$};
			\node[state, right of = qi] (q1) {$q_1$};
			\node[state, right of = q1] (q2) {$q_2$};
			\node[state, accepting, right of = q2] (qf) {$q_3$};
			
			\path[->]
				(qi) edge (q1)
				(q1) edge [loop above] node[above] {$u$} (q1)
				(q1) edge node[above] {$u$} (q2)
				(q2) edge [loop above] node[above] {$u$} (q2)
				(q2) edge (qf)
			;
		\end{tikzpicture}
	}
}

\newcommand{\drawDumbbell}{
	\scalebox{\figureScale}{
		\begin{tikzpicture}[>=stealth, node distance=2cm, thick]
			\node[state, initial, initial text=] (qi) {$s_0$};
			\node[state, right of = qi] (q1) {$s_1$};
			\node[state, right of = q1] (q2) {$s_2$};
			\node[state, accepting, right of = q2] (qf) {$s_3$};
			
			\path[->]
				(qi) edge (q1)
				(q1) edge [loop above] node[above] {$
					\begin{array}{c|c}
						u & v_1
					\end{array}$} (q1)
				(q1) edge node[above] {$
					\begin{array}{c|c}
						u & v_2
					\end{array}$} (q2)
				(q2) edge [loop above] node[above] {$
					\begin{array}{c|c}
						u & v_3
					\end{array}$} (q2)
				(q2) edge (qf)
			;
		\end{tikzpicture}
	}
}

\newcommand{\drawCoTerminalCircuits}{
	\scalebox{\figureScale}{
		\begin{tikzpicture}[>=stealth, node distance=2cm, thick]
			\node[state, initial, initial text=] (qi) {$q_0$};
			\node[state, above right of = qi] (q1) {$q_1$};
			\node[state, below right of = qi] (q2) {$q_2$};
			\node[state, accepting, above right of = q2] (qf) {$q_3$};
			
			\path[->]
				(qi) edge (q1)
				(qi) edge (q2)
				(q1) edge [loop above] node[above] {$
					\begin{array}{c|c}
					u & v_1
					\end{array}$} (q1)
				(q2) edge [loop below] node[below] {$
					\begin{array}{c|c}
					u & v_2
					\end{array}$} (q2)
				(q1) edge (qf)
				(q2) edge (qf)
			;
		\end{tikzpicture}
	}
}

\newcommand{\drawW}{
	\scalebox{\figureScale}{
		\begin{tikzpicture}[>=stealth, node distance=2cm, thick]
			\node[state, initial, initial text=] (qi) {$r_0$};
			\node[state, above right of = qi] (q1) {$r_1$};
			\node[state, above left of = q1] (q2) {$r_2$};
			\node[state, above right of = q1] (q3) {$r_3$};
			\node[state, below right of = q3] (q4) {$r_4$};
			\node[state, above right of = q4] (q5) {$r_5$};
			\node[state, accepting, below right of = q4] (qf) {$r_6$};
			
			\path[->]
				(qi) edge (q1)
				(q1) edge[bend left] node[below left] {$
					\begin{array}{c|c}
						u_1 & v_1
					\end{array}$} (q2)
				(q2) edge [loop above] node[above] {$u_2$} (q2)
				(q2) edge[bend left] node[above right] {$u_3$} (q1)
				(q1) edge node[below right] {$u_1$} (q3)
				(q3) edge [loop above] node[above] {$
					\begin{array}{c|c}
						u_2 & v_2
					\end{array}$} (q3)
				(q3) edge node[below left] {$u_3$} (q4)
				(q4) edge[bend left] node[above left] {$u_1$} (q5)
				(q5) edge [loop above] node[above] {$u_2$} (q5)
				(q5) edge[bend left] node[below right] {$u_3$} (q4)
				(q4) edge (qf)
			;
		\end{tikzpicture}
	}
}

\newcommand{\drawExpAmbiguity}{
	\scalebox{\figureScale}{
		\begin{tikzpicture}[>=stealth, node distance=2cm, thick]
		\node[state, initial, initial text=] (qi) {$q_0$};
		\node[state, right of = qi] (q1) {$q_1$};
		\node[state, accepting, right of = q1] (qf) {$q_2$};
		
		\path[->]
		(qi) edge (q1)
		(q1) edge [loop above] node[above] {$u$} (q1)
		(q1) edge [loop below] node[below] {$u$} (q1)
		(q1) edge (qf)
		;
		\end{tikzpicture}
	}
}

\newcommand{\drawPathEqual}{
	\scalebox{\figureScale}{
		\begin{tikzpicture}[>=stealth, node distance=2cm, thick]
			\node[state, accepting, initial, initial text=$B_{\pi_i=\pi_j}$] (qi) {};
			
			\path[->]
				(qi) edge [loop right] node[right] {$\{ (\sigma_1,\dots,\sigma_n) \mid\sigma_i = \sigma_j\}$} (qi)
			;
		\end{tikzpicture}
	}
}

\newcommand{\drawPrefix}{
	\scalebox{\figureScale}{
		\begin{tikzpicture}[>=stealth, node distance=2cm, thick]
			\node[state, accepting, initial, initial text=$B_{\pi_i \prefix_I \pi_j}$] (qi) {};
			
			\path[->]
				(qi) edge [loop right] node[right] {$\{ (\sigma_1,\dots,\sigma_n) \mid \sigma_i \in \Lambda \implies \sigma_i = \sigma_j \}$} (qi)
			;
		\end{tikzpicture}
	}
}

\newcommand{\drawSize}{
	\scalebox{\figureScale}{
		\begin{tikzpicture}[>=stealth, node distance=2cm, thick]
		\node[state, accepting, initial, initial text=$B_{\pi_i \leq_I^{\text{len}} \pi_j}$] (qi) {};
		
		\path[->]
			(qi) edge [loop right] node[right] {$\{ (\sigma_1,\dots,\sigma_n) \mid \sigma_j = \bot \implies \sigma_i = \bot \}$} (qi)
		;
		\end{tikzpicture}
	}
}

\newcommand{\drawStateLeftInit}{
	\scalebox{\figureScale}{
		\begin{tikzpicture}[>=stealth, node distance=5.2cm, thick]
		\node[state, initial, initial text=$B_{\init(\pi_i^\vartriangleleft)}$] (q0) {};
		\node[state, accepting, right of = q0] (q1) {};
		
		\path[->]
			(q0) edge node[above] {$\{ (q_1, \dots, q_n)
                          \in Q \mid q_i\in Q_0\}$} (q1)
			(q1) edge [loop right] node[right] {$\Sigma^n$} (q1)
		;
		\end{tikzpicture}
	}
}

\newcommand{\drawStateLeftFinal}{
	\scalebox{\figureScale}{
		\begin{tikzpicture}[>=stealth, node distance=5.2cm, thick]
		\node[state, initial, initial text=$B_{\init(\pi_i^\vartriangleleft)}$] (q0) {};
		\node[state, accepting, right of = q0] (q1) {};
		
		\path[->]
			(q0) edge node[above] {$\{ (q_1, \dots, q_n)
                          \in Q \mid q_i\in Q_F\}$} (q1)
			(q1) edge [loop right] node[right] {$\Sigma^n$} (q1)
		;
		\end{tikzpicture}
	}
}

\newcommand{\drawStateRightInit}{
	\scalebox{\figureScale}{
		\begin{tikzpicture}[>=stealth, node distance=6cm, thick]
		\node[state, initial, initial text=$B_{\init(\pi_i^\vartriangleright)}$] (q0) {};
		\node[state, accepting, right of = q0] (q1) {};
		
		\path[->]
			(q0) edge [loop below] node[below] {$\Sigma^n$} (q0)
			(q0) edge node[above] {$\{ (q_1, \dots, q_n)
                          \in Q \mid q_i\in Q_0\}$} (q1)
			(q1) edge [loop below] node[below] {$\{
                          (\sigma_1,\dots,\sigma_n)\in \Sigma^n\mid \sigma_i=\bot\}$} (q1)
		;
		\end{tikzpicture}
	}
}

\newcommand{\drawStateRightFinal}{
	\scalebox{\figureScale}{
		\begin{tikzpicture}[>=stealth, node distance=6cm, thick]
		\node[state, initial, initial text=$B_{\init(\pi_i^\vartriangleright)}$] (q0) {};
		\node[state, accepting, right of = q0] (q1) {};
		
		\path[->]
			(q0) edge [loop below] node[below] {$\Sigma^n$} (q0)
			(q0) edge node[above] {$\{ (q_1, \dots, q_n)
                          \in Q \mid q_i\in Q_F\}$} (q1)
			(q1) edge [loop below] node[below] {$\{
                          (\sigma_1,\dots,\sigma_n)\in \Sigma^n\mid \sigma_i=\bot\}$} (q1)
		;
		\end{tikzpicture}
	}
}

\newcommand{\drawStateRight}{
	\scalebox{\figureScale}{
		\begin{tikzpicture}[>=stealth, node distance=6cm, thick]
		\node[state, initial, initial text=$B_{\pi_x^\vartriangleright = q}$] (q0) {};
		\node[state, right of = q0] (q1) {};
		\node[state, accepting, right of = q1] (q2) {};
		
		\path[->]
			(q0) edge [loop below] node[below] {$(\Sigma_\bot)^n$} (q0)
			(q0) edge node[above] {$\{ (q_1, \dots, q_n) \in Q \cup \{\bot\} \mid q_x = q \}$} (q1)
			(q1) edge node[above] {$\{ (a_1, \dots, a_n) \in \Sigma\cup\{\bot\} \mid a_x = \bot\}$} (q2)
			(q2) edge [loop below] node[below] {$(\Sigma_\bot)^n$} (q2)
		;
		\end{tikzpicture}
	}
}

\newcommand{\drawNP}{
	\scalebox{\figureScale}{
		\begin{tikzpicture}[>=stealth, node distance=2cm, thick]
		\node[state, initial, initial text=] (q0) {};
		\node[state, right of = q0] (q1) {};
		\node[state, right of = q1] (q) {};
		\node[state, accepting, right of = q] (qn) {};
		
		\path[->]
			(q0) edge[bend left] node[above] {$x_1$} (q1)
			(q0) edge[bend right] node[below] {$-x_1$} (q1)
			(q1) edge[dashed, bend left] (q)
			(q1) edge[dashed, bend right] (q)
			(q) edge[bend left] node[above] {$x_k$} (qn)
			(q) edge[bend right] node[below] {$-x_k$} (qn)
			(qn) edge[loop above] node[above] {$0$} (qn)
		;
		\end{tikzpicture}
	}
}

	\vspace{-7mm}
\section{Introduction}\label{sec:introduction}
\vspace{-2mm}

\paragraph{Motivations} An important aspect of automata theory is the
definition of automata
subclasses with particular properties, of algorithmic interest for
instance. As an example, the inclusion problem for
non-deterministic finite automata is \textsc{PSpace-c} but becomes
\textsc{PTime} if the automata are $k$-ambiguous for a fixed
$k$~\cite{SICOMP::StearnsH1985}.

By automata
theory, we mean automata in the general sense of finite state
machines processing finite words. This includes what we call \emph{automata with outputs},
which may also produce output values in a fixed monoid $\M =
(D,\oplus,\zero)$. In such an automaton, the transitions are extended
with an (output) value in $D$, and the value of an accepting path is
the sum (for $\oplus$) of all the values occurring along its
transitions. Automata over finite words in $\Lambda^*$ and with outputs
in $\M$ define subsets of $\Lambda^*\times D$ as follows: to any input
word $w\in\Lambda^*$, we associate the set of values of all the
accepting paths on $w$. For example,
transducers are automata with outputs in a free monoid: they process
input words and produce output words and therefore define binary
relations of finite words~\cite{Bers79}.

The many decidability
properties of finite automata do not carry over to transducers, and
many restrictions have been defined in the literature to recover
decidability, or just to define subclasses relevant to particular
applications. The inclusion problem for transducer is
undecidable~\cite{DBLP:journals/jacm/Griffiths68}, but decidable for \emph{finite-valued}
transducers~\cite{SICOMP::Weber1993}. Another well-known subclass is that of 
the \emph{determinisable} transducers~\cite{BealCPS03}, defining
sequential functions of words. Finite-valuedness and determinisability
are two properties decidable in \textsc{PTime}, i.e., it is
decidable in \textsc{PTime}, given a transducer, whether it is
finite-valued (resp. determinisable). As a second example of automata with outputs, we also consider
sum-automata, i.e. automata with outputs in $(\Z,+,0)$, which defines
relations from words to $\Z$. Properties such as
functionality, determinisability, and $k$-valuedness (for a fixed $k$)
are decidable in \textsc{PTime} for
sum-automata~\cite{DBLP:journals/corr/abs-1111-0862,DBLP:conf/fsttcs/FiliotGR14}.

In our experience, it is quite often the case that deciding a subclass
goes in two steps: $(1)$ define a characterisation of the subclass
through a ``simple'' pattern, $(2)$ show how to decide the existence of a
such a pattern. For instance, the determinisable transducers have been
characterised via the so called \emph{twinning property}~\cite{DBLP:journals/tcs/Choffrut77,DBLP:journals/iandc/WeberK95,DBLP:journals/tcs/BealC02}, which, said
briefly, asks that the output words produced by any two different paths
on input words of the form $uv^n$ cannot differ unboundedly when $n$
grows, with a suitable definition of ``differ''. Quite often, the most difficult part is
step $(1)$ and step $(2)$ is technical but less difficult to
achieve, as long as we do not seek for optimal complexity
bounds (by this we mean that \textsc{PTime} is good enough, and
obtaining the best polynomial degree is not the objective). We even
noticed that in transducer theory, even though step $(2)$ share common
techniques (reduction to emptiness of reversal-bounded counter
machines for instance), the algorithms are often ad-hoc to the
particular subclass considered. Here is a non-exhaustive list of subclasses of
transducers which are decidable in \textsc{PTime}: determinisable
transducers~\cite{DBLP:journals/tcs/Choffrut77,DBLP:journals/iandc/WeberK95,BealCPS03,DBLP:journals/tcs/BealC02,DBLP:journals/jalc/AllauzenM03,ChoffrutDecomposition}, functional transducers~\cite{BealCPS03,DBLP:journals/tcs/BealC02}, $k$-sequential
transducers (for a fixed $k$)~\cite{DBLP:conf/fossacs/DaviaudJRV17}, multi-sequential
transducers~\cite{DBLP:journals/ijfcs/JeckerF18,ChoffrutDecomposition}, $k$-valued transducers (for a fixed $k$)~\cite{GurIba83},
finite-valued transducers~\cite{DBLP:conf/mfcs/SakarovitchS08,SICOMP::Weber1993}. Our goal in
this paper is to define a common tool for step $(2)$, i.e., define a
generic way of deciding a subclass characterised through a
structural pattern. More precisely, we want to define logics, 
tailored to particular monoids $\M$, able to express properties of automata with outputs in $\M$, such that
model-checking these properties on given automata can be done in
\textsc{PTime}.

\vspace{-2mm}
\paragraph{Contributions} We define a general logic, denoted
$\logic[\O]$ for ``pattern logic'',  to express properties of automata with outputs in a fixed monoid $\M = (D,\oplus,\zero)$. This
logic is parameterised by a set of predicates $\O$ interpreted on
$D$. We first give sufficient conditions under
which the problem of model-checking an
automaton with outputs in $\M$ against
a formula in this logic is decidable. Briefly, these conditions
require the existence of a machine model accepting tuples of runs 
which satisfy the atomic predicates of the logic, is closed under
union and intersection, and has decidable emptiness problem.

Then, we study three particular classes of automata with outputs: finite
automata (which can be seen as automata with outputs in a trivial
monoid with a single element), transducers (automata with outputs in a
free monoid), and sum-automata (automata with outputs in
$(\Z,+,0)$). For each of them, we define particular logics, called
$\logic_\nfa$, $\logic_\trans$ and $\logic_\suma$ to express
properties of automata with outputs in these particular
monoids. Formulas in these logics have the following form:
$$
\exists \pi_1:p_1\xrightarrow{u_1\mid v_1} q_1,\dots,\exists
\pi_n:p_n\xrightarrow{u_n\mid v_n} q_n, \C
$$
where the $\pi_i$ are path variables, the $p_i,q_i$ are state
variables, the $u_i$ are (input) word variables and the $v_i$ are
output value variables (interpreted in $D$). The subformula $\C$ is a
quantifier free Boolean combinations of constraints talking about
states, paths, input words and output values. Such a formula expresses
the fact that there exists a path $\pi_1$ from some state $p_1$ to some
state $q_1$, over some input word $u_1$, producing some value $v_1$,
some path $\pi_2$ etc. such that they all satisfy the constraints in
$\C$. In the three logics, paths can be tested for equality. Input
words can be compared with the prefix relation, w.r.t. their
length, and their membership to a regular language be
tested. States can be compared for equality, and it can be
expressed whether they are initial or final.

The predicates we take for the output values depends on the monoids. For transducers, output words can be compared
with the non-prefix relation (and by derivation $\neq$), a predicate
which cannot be negated (otherwise
model-checking becomes undecidable), and can also be compared with respect to their length, and
membership to a regular language can be tested. For sum-automata, the
output values can be compared with $<$ (and by derivation
$=,\neq,\leq$). As an example, a transducer (resp. sum-automaton) is
\emph{not} $(n-1)$-valued iff it satisfies the following
$\logic_\trans$-formula (resp. $\logic_\suma$-formula):
$$
\exists \pi_1:p_1\xrightarrow{u\mid v_1} q_1,\dots,\exists
\pi_{n}:p_{n}\xrightarrow{u\mid v_n} q_n, \bigwedge_{i=1}^n
\init(p_i)\wedge \final(q_i)\wedge \bigwedge_{1\leq i<j\leq n} v_i\neq v_j.
$$
For the three logics, we show that
deciding whether a given automaton
satisfies a given formula is \textsc{PSPace-c}. When the formula is
\emph{fixed}, the model-checking problem becomes \textsc{NLogSpace-c}
for $\logic_\nfa$ and $\logic_\trans$, and \textsc{NP-c} for $\logic_\suma$. If
output values can only be compared via disequality $\neq$ (which
cannot be negated), then $\logic_\suma$ admits \textsc{PTime}
model-checking. We show that many of the properties from the literature, including all
the properties mentioned before, can be expressed in these logics. 
As a consequence, we show that most of the \textsc{PTime} upper-bounds obtained 
for deciding subclasses of finite automata in \cite{DBLP:journals/tcs/WeberS91,DBLP:journals/ijfcs/AllauzenMR11}, of transducers in
\cite{DBLP:journals/tcs/Choffrut77,GurIba83,DBLP:journals/iandc/WeberK95,DBLP:journals/acta/Weber89,DBLP:journals/ijfcs/JeckerF18,ChoffrutDecomposition,BealCPS03,DBLP:conf/mfcs/SakarovitchS08,DBLP:conf/fossacs/DaviaudJRV17} and sum-automata in \cite{DBLP:journals/corr/abs-1111-0862,DBLP:conf/fsttcs/FiliotGR14,DBLP:conf/fossacs/DaviaudJRV17,conf/lata/BalaK13}, can be directly obtained by
expressing in our logics the structural patterns given in these papers, which
characterise these subclasses.

\vspace{-1mm}
\paragraph{Related works} In addition to the results already
mentioned, we point out that the syntax of
our logic is close to a logic, defined in \cite{DBLP:conf/lics/FigueiraL15} by Figueira
and Libkin,  to express path queries in graph databases (finite graphs
with edges labelled by a symbol). In this work, there is no
disjunction nor negation, and no distinction between input and output
values. By making such a distinction, and by adding negation and
disjunction, we were able to tailor our logics to particular automata
models and add enough power to be able to directly express classical structural automata properties.


\vspace{-2mm}
	\section{Finite Automata with Outputs}\label{sec:preliminaries}

In this section, we define a general model of finite automata
defining functions from the free monoid $\Lambda^*$ (where
$\Lambda$ is a finite input alphabet) to any monoids $\M = (D,\oplus,\zero)$.
More precisely, they are parametrised by a monoid of
\emph{output values}, read input words over some alphabet and output
elements of the output monoid, obtained by summing the output
values met along accepting paths. 

Formally,  a \emph{monoid} $\M$ is a tuple $(D,\oplus_\M,\zero_\M)$ where $D$ is a set
of elements which we call here values or sometimes outputs,
$\oplus_\M$ is an associative binary operation on $D$, for which
$\zero_\M\in D$ is neutral. Monoids of interest in this paper are the free monoid $(\Lambda^*,\cdot,\varepsilon)$ for some finite
alphabet of symbols $\Lambda$ (where $\cdot$ denotes the
concatenation), and the monoid $(\Z,+,0)$. 
We also let $\Lambda_\varepsilon = \Lambda\cup
\{\varepsilon\}$. 	For $w \in \Lambda^*$, $|w|$ denotes its length, in particular
$|\varepsilon| = 0$. The set of positions of $w$ is $\{1,\dots,|w|\}$
(and empty if $w = \epsilon$). We let $w[i]$ be the $i$th symbol of
$w$. Given $w_1,w_2$, we write $w_1\prefix
w_2$ whenever $w_1$ is a prefix of $w_2$. \emph{All over this paper, the input alphabet is denoted by the letter $\Lambda$}.


\begin{definition}[Automata with outputs]
	An \emph{automaton $A$ with outputs} over an (output) monoid 
        $\M=(D, \oplus_\M, \zero_\M)$ is a tuple $\langle Q, I,F, \Delta, \gamma \rangle$ where
	$Q$ is a non-empty finite set of states, $I\subseteq Q$
                the set of initial states, $F\subseteq Q$ the
                set of final states, 
	$\Delta \subseteq Q \times \Lambda_\varepsilon
                \times Q$ the set of transitions labelled
                with some element of $\Lambda_\varepsilon$,
	and $\gamma\colon \Delta \rightarrow D$ a
                mapping from transitions to output values\footnote{Often in the literature,
                  output values are directly given in the
                  transitions, i.e. the transition relation is a
                  (finite) subset of $Q\times
                  \Lambda_\varepsilon\times D\times Q$.
                  Our definition is then equivalent modulo \textsc{PTime} transformation,
                and allows for a clearer distinction between input and
                output mechanisms.}. The set of automata
                over $\M$ is written
                $\trs(\M)$.
\end{definition}

We write $\dim(A)$ to refer to the number of states of $A$.
A \emph{path} in $A$ is a sequence $\pi = q_0 a_1 d_1 q_1 \dots a_n d_n q_n \in Q(\Lambda_\varepsilon D Q)^*$, for $n\geq 0$, such that for all $1 \leq i \leq n$ we have $(q_{i-1}, a_i, q_i) \in \Delta$ and $\gamma(q_{i-1}, a_i, q_i) = d_i$.
The \emph{input} of $\pi$ is defined as the word $\inp(\pi) =
a_1\dots a_n$ (and $\varepsilon$ if $\pi\in Q$), the \emph{output} of
$\pi$ as the element $\out(\pi) = d_1\oplus_\M \dots \oplus_\M d_n$
(and $\zero_\M$ if $\pi\in Q$), and the \emph{size} of $\pi$ as $|\pi| = n$.
We may write $\pi : q_0\xrightarrow{\inp(\pi)\mid \out(\pi)} q_n$ to denote that $\pi$ is a path from $q_0$ to
$q_n$ on input $\inp(\pi)$ and output $\out(\pi)$.
For convenience we write $\pi^\first, \pi^\last$ to denote respectively the starting state $q_0$ and the ending state $q_n$ of the path $\pi$.
The set of all paths of $A$ is written $\paths(A)$.
A path $\pi : q_0\xrightarrow{u\mid v} q_n$ is \emph{initial} if $q_0\in I$, \emph{final} if $q_n\in F$ and accepting if it is both initial and final.
The set of accepting paths of $A$ is denoted by $\paths_{acc}(A)$. The \emph{input/output relation}
(or just relation) defined by $A$ is the set of
pairs $R(A) \subseteq \Lambda^*\times D$ defined by 
$$
R(A) = \{ (u,v)\mid \exists \pi\in
\paths_{acc}(A)\cdot \inp(\pi) = u\wedge \out(\pi) = v\}
$$



\paragraph{Finite automata, transducers and sum-automata}

In this paper, we consider three instances of
automata with outputs. First, \emph{finite automata} (over $\Lambda$), are
seen as automata with outputs in a trivial monoid (and which is therefore
ignored). \emph{Transducers}  are automata with outputs in the free 
monoid $\Gamma^*$. They define relations from
$\Lambda^*$ to $\Gamma^*$. Finally, \emph{sum-automata} are
automata with outputs in the monoid $(\mathbb{Z},+,0)$.


	\vspace{-3mm}
\section{A Pattern Logic for Automata with Outputs}\label{sec:logic}
\vspace{-1mm}

In this section, we introduce a generic pattern logic. It is built over four kind of
variables, namely path, state, input and output
variables. More precisely, we let 
$\varpaths = \{ \pi,\pi_1,\dots\}$, $\varstates = \{
q,q_1,p\dots,\}$, $\varin = \{ u,u_1,\dots\}$ and $\varout =
\{ v,v_1,\dots\}$ be disjoint and countable sets of resp. path,
state, input and output variables. We define $Terms(\varout,\oplus,\zero)$ as the set of terms
built over variables of $\varout$, a binary function symbol $\oplus$
(representing the monoid operation) and constant symbol $\zero$
(neutral element).

The logic syntax is parametrised by a set of output
predicates $\O$. Output predicates of arity $0$ are called constant
symbols, and we denote by $\O|_n$ the predicates of arity $n$. 
Predicates talking about states, paths and input words are however fixed in
the logic.


\begin{definition}
	A \emph{pattern formula} $\varphi$ over a set of output predicates $\O$ is of the form 
    $$
    	\varphi\ =\
    		\exists \pi_1\colon p_1\xrightarrow{u_1\mid v_1} q_1,
    		\dots, 
    		\exists \pi_n\colon p_n\xrightarrow{u_n\mid v_n} q_n,
    		\C
    $$
    where for all $1\leq i\leq n$, $\pi_i\in \varpaths$
    and they are all pairwise different,
    $p_i,q_i\in\varstates$, $u_i\in \varin$,
    $v_i\in\varout$, and $\mathcal{C}$ is a Boolean
    combination of atoms amongst
	$$
		\begin{array}{lcl@{\quad}r}
			\text{Input constraints} & \colon &
				u \prefix u' \mid u \in L \mid |u| \leq |u'|
				& u,u'\in \varin
		\smallskip\\
			\text{Output constraints} & \colon &
				p(t_1,\dots,t_n)
				& p \in \O|_n, t_{i} \in Terms(\varout, \oplus,\zero)
		\smallskip\\
			\text{State constraints} & \colon &
				\init(q) \mid \final(q)\mid q = q'
				& q,q'\in \varstates
		\smallskip\\
			\text{Path constraints} & \colon &
				\pi = \pi'
				& \pi,\pi'\in \varpaths
		\end{array}
	$$
	where $L$ is a regular language of words over $\Lambda$ (assumed to be represented as an NFA).
	The sequence of existential quantifiers before $\C$ in $\varphi$ is called the \emph{prefix} of $\varphi$. 
	We denote by $\logic(\O)$ the set of pattern formulas over
        $\O$, and by $\logic^+(\O)$ the fragment where output
        predicates does not occur under an odd number of negations. 
\end{definition}

The size of a formula is the number of its symbols plus the number of
states of all NFA representing the membership constraints. 
We denote by $\Var(\varphi)$ the variables occurring in any pattern formula $\varphi$, and
by $\Var_P(\varphi)$ (resp. $\Var_Q(\varphi)$,
$\Var_I(\varphi)$, $\Var_O(\varphi)$) its restriction to path
(resp. state, input, output) variables. 
We finally let
$(u = u') \new u\prefix u' \land  u'\prefix u$,
$(|u|= |u'|) \new (|u|\leq |u'|)\land (|u'|\leq |u|)$,
$(|u|< |u'|) \new \lnot (|u'|\leq |u|)$.

\paragraph{Semantics}
To define the semantics of a pattern formula $\varphi$, we first fix
some monoid $\M = (D,\oplus_\M,\zero_\M)$ together with an
interpretation ${p}^\M$ of each output predicates $p \in \O$ of
arity $\alpha(p)$, such that $p^\M \in D$ if $p$ is a constant and
$p^\M \subseteq D^{\alpha(p)}$ otherwise. 
Given a valuation $\nu\colon \varout \rightarrow D$, the
interpretation $.^\M$ 
can be inductively extended to terms $t$ by letting $\zero^{\nu, \M} = \zero_\M$,
$(t_1\oplus t_2)^{\nu, \M} = t_1^{\nu, \M}\oplus_\M t_2^{\nu, \M}$ and
$x^{\nu, \M} = \nu(x)$.

Then, a formula $\varphi\in\logic(\O)$ is interpreted in an automaton
with outputs $A\in\trs(\M)$ as a set of valuations $\sem{\varphi}_A$ of $\Var(\varphi)$ which we now define.
Each valuation $\nu\in\sem{\varphi}_A$ maps state variables to states of $A$, path variables to paths of $A$, etc.
Such a valuation $\nu$ satisfies an atom $u\prefix u'$ if $\nu(u)$ is a prefix of $\nu(u')$, $u\in L$ if $\nu(u)\in L$, $|u|\leq|u'|$ if $|\nu(u)|\leq|\nu(u')|$.
Given a predicate $p \in \O$ of arity $\alpha(p)$, an atom $p(t_1,\dots,t_{\alpha(p)})$ is satisfied by $\nu$
if $(t_1^{\nu,\M},\dots,t_{\alpha(p)}^{\nu,\M})\in p^\M$.
Finally, $\nu$ satisfies $\init(q)$ (resp. $\final(q)$) if $\nu(q)$ is
initial (resp. $\nu(q)$ is final). The
satisfiability relation is naturally extended to Boolean
combinations of atoms. 
Finally, assume that $\varphi$ is of the form
$
	\exists \pi_1\colon p_1\xrightarrow{u_1\mid v_1} q_1,
	\dots,
	\exists \pi_n \colon p_n\xrightarrow{u_n\mid v_n} q_n,
	\C
$,
we say that $A$ satisfies $\varphi$, denoted by $A \models \varphi$, if there exists a valuation $\nu$ of $\Var(\varphi)$ such that for all $i\in\{1,\dots,n\}$, $\nu(\pi_i) \colon \nu(p_i)\xrightarrow{\nu(u_i)\mid
\nu(v_i)}\nu(q_i)$ and $\nu$ satisfies $\C$ ($\nu\models \C$).
Given a pattern formula
$\varphi$ and an automaton with outputs $A$, the \emph{model-checking problem} consists in deciding
whether $A$ satisfies $\varphi$, i.e. $A \models \varphi$.


\begin{example}\label{ex:logicex} Given
$k\in\N$, the $k$-valuedness property has been already
expressed in Introduction (assuming ${=}\in\O$). The formula $
    	\exists \pi_0\colon p_0\xrightarrow{u|v_0} q_0,
    	\dots,
    	\exists \pi_k \colon p_k\xrightarrow{u|v_{k}} q_k,\C_0$ where
        $\C_0 = \bigwedge_{0\leq i<j\leq k} \pi_i\neq \pi_j\land \bigwedge_{i=0}^k \init(p_i)\land \final(q_i)$
expresses the fact that an automaton is not $(k-1)$-ambiguous (has at least
$k$ accepting paths for some input). 
\end{example}


	\section{Model-Checking Problem}\label{sec:MC}

In this section, we give sufficient conditions on the output monoid $\M$ and
the set of output predicates $\O$ by which the model-checking of automata
with outputs in $\M$ against pattern formulas over the output
predicates $\O$ is decidable. In the next sections, we study the
precise complexity of the model-checking problem for particular
monoids $\M$.

\paragraph{Tuple acceptors} Since automata with outputs can get their output values in arbitrary
monoids, to get an effective model-checking algorithm, we will assume the existence
of machines, called tuple acceptors, that can recognise
sets of word tuples. These machines will be required to satisfy some key
properties, forming the notion of \emph{good class} of tuple acceptors. First, what we
call a \emph{tuple acceptor} is a machine $M$ whose semantics is a set of
tuples of words $\sem{M}\subseteq (\Sigma^*)^n$, for some alphabet
$\Sigma$ and some arity $n\geq 1$. The notion of good class, formally
defined later, require
$(i)$ that any regular set of tuples is recognised by some machine, for a
regularity notion that we will make clear (roughly, by seeing tuples of
words as words resulting from the overlapping of all components),
$(ii)$ all output predicates (and their negation) are recognised by some machine, $(iii)$
the class is closed under union and intersection.

\paragraph{Regular sets of word tuples} Let $\Sigma$ be some
alphabet containing some symbol $\bot$, $\pi\in \Sigma^*$ and $m\geq |\pi|$. The \emph{padding} of $\pi$
with respect to $m$ is the word $\pi' = \pi \bot^{m-|\pi|}$. Let $\pi_1,\pi_2\in \Sigma^*$ and let $m =
\text{max}(|\pi_1|,|\pi_2|)$. For $j=1,2$, let $\pi'_j$ the padding of
$\pi_j$ with respect to $m$. Note that $|\pi'_1| = |\pi'_2| = m$. 
The \emph{convolution} $\pi_1\conv \pi_2$ is the word of length $m$
defined for all $1\leq i\leq n$ by $(\pi_1\conv \pi_2)[i] =
(\pi'_1[i],\pi'_2[i])$. E.g. $q_1\lambda_1d_1q_2\otimes p_1 =
(q_1,p_1)(\lambda_1,\bot)(d_1,\bot)(q_2,\bot)$. 
The convolution can be naturally extended to multiple words as follows:
$\bigconv_{i=1}^n \pi_i = \pi_1\conv (\pi_2\conv \dots \conv\pi_n)$.
\begin{definition}
A set of $n$-ary word tuples $P\subseteq (\Sigma^*)^n$ is 
\emph{regular} if $L = \{ \bigconv_{i=1}^n \pi_i
\mid (\pi_1, \dots, \pi_n)\in P \}$ is a regular language over 
$\Sigma^n$. We often identify $L$ and $P$. 
\end{definition}


\paragraph{Good class of tuple acceptors} First, any valuation $\nu$ of a
set of path variables $X$ into paths of some automaton with values in some
monoid $\M$ gives a way to interpret terms $t\in \text{Terms}(X, \oplus,
\zero)$ as follows: for $\pi\in X$, $\pi^{\nu,\M} = \out(\nu(\pi))$,
$\zero^{\nu,\M} = \zero_\M$ and $(t_1\oplus t_2)^{\nu,\M} =
t_1^{\nu,\M}\oplus_\M t_2^{\nu,\M}$. Then, for a class $\C$ (i.e. a
set) of tuple acceptors, we denote by $\C|_n$ its restriction to
acceptors of arity $n$.

\begin{definition}[Good class]
A class of tuple acceptors $\C$ is said to be
\emph{good} for an output monoid $\M = (D,\oplus_\M,\zero_\M)$, a set of output predicates $\O$ and
an interpretation $p^\M\subseteq D^{\alpha(p)}$ for all $p\in\O$ of
arity $\alpha(p)$, if the following conditions are satisfied:
\begin{enumerate}

\item for all automata with outputs $A\in \trs(\M)$ with a set of
  states $Q$ we have:
	\begin{enumerate}
	
	\item \label{goodreg}
		$\forall n\geq 1$,$\forall R\subseteq \paths(A)^n$ regular,
	        $R = \sem{M}$ for some $M\in\C|_n$.
	\item
	  all $p\in \O$ of arity $\alpha(p)$,  
		all $X = \{ \pi_1, \dots, \pi_n \}$ finite sets of path variables
		and all $t_1, \dots, t_{\alpha(p)} \in \text{Terms}(X, \oplus, \zero)$, 
		there exist $M,M'\in\C|_{n}$ such that
                \begin{enumerate}
                    \item\label{goodout}
                    $\sem{M} = \{ 
                     \left( \nu(\pi_1), \dots, \nu(\pi_n) \right) |
                     \nu \colon X\rightarrow \paths(A) \land
                     (t^{\nu, \M}_1, \dots, t^{\nu, \M}_{\alpha(p)}) \in p^\M
                  \}$
                    \item\label{goodneg}  $\sem{M'} = \paths(A)^n\setminus \sem{M}$.

                \end{enumerate}
	\end{enumerate}

\item \label{goodclos}
	$\forall n\geq 1$, $\forall M_1,M_2\in\C|_n$, there exist $M, M'\in\C|_n$ such that
	$\sem{M}=\sem{M_1}\cap \sem{M_2}$ and  $\sem{M'}=\sem{M_1}\cup \sem{M_2}$.
\end{enumerate}
We say that $\C$ is \emph{effective} if all properties are effective 
and moreover it is decidable whether $\sem{M}\neq\varnothing$ for any (effectively represented) $M\in\C$. 
We say that $\C$ is \emph{weakly good} if all properties hold except \ref{goodneg}.
\end{definition}


	


Effectiveness of a good class gives effective model-checking, as
announced.

\begin{theorem}\label{thm:algo}
    Let $\M$ be a monoid and $\O$ be a set of output predicates,
    interpreted over $\M$. If there exists an effective good class
    $\C$ (resp. effective weakly good class) of tuple
    acceptors for $\M$ and $\O$, then the
    model-checking problem of automata with outputs in $\M$
    against pattern formulas $\psi\in\logic[\O]$
    (resp. $\psi\in\logic^+[\O]$) is decidable.
\end{theorem}

\begin{proof}[sketch]
    First, the formula is put in negation normal form: negation is
    pushed down to the atoms. Then, given an automaton with outputs in $\M$, we show that any
    tuple of paths which satisfy state, input and path predicates and
    their negations is a regular set of
    path tuples (this is doable even for input equality as well as input
    length comparison thanks to the way paths are overlapped by the definition of
    convolution). By condition \ref{goodreg}, these sets of tuples are
    accepted by acceptors of $\C$. By conditions \ref{goodout} and
    \ref{goodneg}, tuples of paths satisfying output predicates and
    their negations are also accepted by acceptors of $\C$. Then, the
    closure properties (condition \ref{goodclos}) allows us to
    construct an acceptor for the tuples of paths satisfying the whole
    formula inductively.$\hfill\square$
\end{proof}

	\section{A pattern logic for finite automata}\label{sec:regular}

Finite automata can be seen as automata with outputs in a trivial monoid (with a single element). As the monoid is trivial, there is no need for predicates over it and so we specialize our pattern logic into $\logic_\nfa=\logic[\varnothing]$.


\begin{definition}[Pattern logic for NFA]
	The logic $\logic_\text{NFA}$ is the set of formulas 
	$$
		\begin{array}{lllllll}
			\varphi \Coloneqq &
				\exists \pi_1\colon p_1\xrightarrow{u_1} q_1,
				\dots,
				\exists \pi_n\colon p_n\xrightarrow{u_n} q_n,
				\C
			\smallskip\\
			\C \Coloneqq &
				\lnot\C \mid
				\C \lor \C \mid
				u \prefix u' \mid
				u \in L \mid
				|u| \leq |u'| \mid
				\init(q) \mid
				\final(q) \mid
				q = q' \mid
				\pi=\pi'
		\end{array}
	$$
	where for all $i\neq j$, $\pi_i\neq \pi_j$, $L$ is a regular
	language over $\Lambda$ (assumed to be represented as an NFA),
	$u,u'\in\{u_1,\dots,u_n\}$, $q,q'\in\{q_1,\dots,q_n\}$ and
	$\pi,\pi'\in\{\pi_1,\dots,\pi_n\}$.
\end{definition}

As a yardstick to measure the expressiveness of $\logic_\nfa$, we have
considered the structural properties of NFA studied in two classical
papers: \cite{DBLP:journals/tcs/WeberS91} by Weber and Seidl and in~\cite{DBLP:journals/ijfcs/AllauzenMR11} by
Allauzen et al. The authors of these two papers give {\sc PTime} membership
algorithms for $k$-ambiguity, finite ambiguity, polynomial ambiguity
and exponential ambiguity (with as applications the approximation of
the entropy of probabilistic automata for example). We refer the
interested readers to these papers for the formal definitions of those
classes. The solutions to these membership problems follow a recurrent
schema: one defines $(1)$ a pattern that identifies the members of the
class and $(2)$ an algorithm to decide if an automaton satisfies the
pattern. The next theorem states that all these membership
problems can be reduced to the model-checking problem of $\logic_\nfa$
using a constant space reduction. The proof of this theorem is
obtained by showing how the patterns identified in~\cite{DBLP:journals/tcs/WeberS91},
can be succinctly and naturally encoded into (fixed) $\logic_\nfa$
formulas. As a corollary, we get that all the class membership
problems are in  {\sc NLogSpace}, using a model-checking algorithm that we
defined below for $\logic_\nfa$.

\begin{theorem}\label{thm:propNFA}
\label{thm:expressLNFA}
The membership problem to the subclasses of $k$-ambiguous, finitely ambiguous, polynomially
ambiguous and exponentially ambiguous NFA can be reduced to the
model-checking problem of $\logic_\nfa$ with constant space
reduction. The obtained formulas are constant (for fixed
$k$). 
\end{theorem}
\begin{proof}
For each membership problem, our reduction copies (in constant space)
the NFA and considers the model-checking for this NFA against a fixed
$\logic_\nfa$ (one for each class).  As illustration, $k$-ambiguity has already been
expressed in Example~\ref{ex:logicex}. %
As a second example,  an automaton is not polynomially ambiguous iff there exists a state
$p$ which is reachable from an initial state, and the source of two
different cycles labelled identically by a word $v$. With
$\logic_\nfa$ this gives:
	$\exists \pi_0  \colon q_0 \xrightarrow{u_1} p,
		\exists \pi_1 \colon p\xrightarrow{u_2} p,
		\exists \pi_2\colon p\xrightarrow{u_2} p,
		\exists \pi_3\colon p\xrightarrow{u_3} q,
		\init(q_0) {\land} \pi_1\neq \pi_2 {\land} \final(q)\hfill\square$
\end{proof}

\noindent The model-checking problem asks if a given NFA $A$ satisfies a given $\logic_\text{NFA}$-formula
$\varphi$.

\begin{theorem}\label{thm:mainNFA}
	The model-checking problem of NFA against formulas in $\logic_\nfa$ is \textsc{PSpace-C}.
	It is in \textsc{NLogSpace-C} if the formula is fixed.
\end{theorem}

\begin{proof}[sketch]
    We use NFA as acceptors for tuples of paths. The algorithm
    presented in the proof of Theorem~\ref{thm:algo} yields an
    exponentially large NFA (and polynomial if the formula is
    fixed). We show that it does not need to be constructed explicitly
    and that a short non-emptiness witness can be searched
    non-deterministically on-the-fly. For \textsc{PSpace}-hardness, we
    notice that the non-emptiness of the intersection of $n$ DFA can
    be easily expressed in $\logic_\nfa$, by seeing the $n$ DFA as
    a disjoint union, and by asking for the existence of
    $n$ different accepting paths over the same input in this union.$\hfill\square$
\end{proof}


\begin{corollary}[of Theorems~\ref{thm:propNFA} and \ref{thm:mainNFA}]\label{cor:propNFA}
The membership problem to the classes of $k$-ambiguous, finitely ambiguous, polynomially
ambiguous and exponentially ambiguous NFA is in \textsc{NLogSpace}.
\end{corollary}
	\vspace{-5mm}
\section{A pattern logic for transducers}\label{sec:transducers}

Transducers are automata with outputs in a free monoid $\M_{Trans} = (\Gamma^*,\cdot,\varepsilon)$ and therefore define subsets of  $\Lambda^*\times \Gamma^*$.
Since our general pattern logic can test for output equalities (by
repeating twice an output variable in the prefix), the
model-checking is easily shown to be undecidable by encoding PCP:
\begin{theorem}\label{thm:undec}
    The model-checking problem of transducers against formulas in
    $\logic[\varnothing]$ is undecidable. 
\end{theorem}

To obtain a decidable logic for transducers, we need to exclude equality tests on the output words in the logic.
However, as we will see, we can instead have inequality test $\neq$ as
long as it is not under an odd number of negations in the formula. We
also allow to test (non) membership of output word concatenations to a
regular language, as well as comparison of output word
concatenations wrt their length. Formally:
\begin{definition}[Pattern logic for transducers]
The logic $\logic_{\trans}$ is the set of formulas of the form
$$
	\begin{array}{lllllll}
	\varphi & \Coloneqq &
		\exists \pi_1\colon p_1\xrightarrow{u_1\mid v_1} q_1,
		\dots,
		\exists \pi_n\colon p_n\xrightarrow{u_n\mid v_n} q_n,
		\C
	\smallskip\\
	\C & \Coloneqq &
		\lnot \C \mid
		\C \lor \C \mid
		u \prefix u' \mid
		u \in L \mid
		|u| \leq |u'|\mid
		\init(q) \mid
		\final(q) \mid
		q = q' \mid
		\pi=\pi' |
	\\
		&&
		t \not\prefix t' \mid
		t \in N \mid
		|t| \leq |t'|
\end{array}
$$
where for all $1\leq i<j\leq n$, $\pi_i\neq \pi_j$ and $v_i\neq v_j$
(no implicit output equality tests), $L$ (resp. $N$) is a regular
language over $\Lambda$ (resp. $\Gamma$), assumed to be represented as an NFA,
$u,u'\in\{u_1,\dots,u_n\}$, $q,q'\in\{q_1,\dots,q_n\}$,
$t,t'\in Terms(\{v_1,\dots,v_n\},\cdot,\epsilon)$, 
$\pi,\pi'\in\{\pi_1,\dots,\pi_n\}$, and $t\not\prefix t'$ does not
occur under an odd number of negations. 
\end{definition}

We define the macros
$t\neq t' \new t\not\prefix t' \lor  t'\not\prefix t$,
$\text{mismatch}(t,t') \new t\not\prefix t' \land  t'\not\prefix t$ and
$$
\text{SDel}_{\neq}(t_1,t'_1,t_2,t'_2)\new (|t'_1|\neq |t'_2|)\ \lor[\
t'_1t'_2\neq \epsilon\land \text{mismatch}(t_1,t_2)]
$$
Let us explain the latter macro. Many properties of transducers are based on the notion of
output delays, by which to compare output words. Formally, for any two
words $v_1,v_2$, $\text{delay}(v_1,v_2) = (\alpha_1,\alpha_2)$ such that
$v_1 = \ell \alpha_1$ and $v_2=\ell \alpha_2$ where $\ell$ is the longest
common prefix of $v_1$ and $v_2$. It can be seen that for any words
$v_1,v'_1,v_2,v'_2$, if we have
$\text{SDel}_{\neq}(v_1,v'_1,v_2,v'_2)$, then $\text{delay}(v_1,v_2)\neq \text{delay}(v_1v'_1,v_2v'_2)$, but
the converse does not hold. But, if $\text{delay}(v_1,v_2)\neq
\text{delay}(v_1v'_1,v_2v'_2)$, then
$\text{SDel}_{\neq}(v_1(v'_1)^i,v'_1,v_2(v'_2)^i,v'_2)$ holds for some
$i\geq 0$. These two facts allows us to express all the known
transducer properties from the literature relying on the notion of
delays. We leave however as open whether our
logic can express a constraint such as $\text{delay}(v_1,v_2)\neq
\text{delay}(v_3,v_4)$.


We review here some of the main transducer subclasses studied in the
literature. We refer the reader to the mentioned references for the
formal definitions. As for the NFA subclasses of the previous section,
deciding them usually goes in two steps: $(1)$ identify a structural
pattern characterising the property, $(2)$ decide whether such as
pattern is satisfied by a given transducer. 
The class of
determinisable transducers are the transducers which define sequential
functions~\cite{DBLP:journals/tcs/Choffrut77,BealCPS03,DBLP:journals/iandc/WeberK95}. The $k$-sequential transducers are the
transducers defining unions of (graphs) of $k$ sequential
functions~\cite{DBLP:conf/fossacs/DaviaudJRV17}. The multi-sequential ones are the union of all
$k$-sequential transducers for all $k$~\cite{DBLP:journals/ijfcs/JeckerF18,ChoffrutDecomposition}. Finally, the $k$-valued
transducers are the transducers for which any input word has at most
$k$ output words~\cite{GurIba83,journals/mst/SakarovitchS10}, and the
finite-valued ones are all the
$k$-valued transducers for all $k$~\cite{DBLP:journals/acta/Weber89,SICOMP::Weber1993,DBLP:conf/mfcs/SakarovitchS08}. 
All these classes, according to the given references,
are decidable in {\sc PTime}.

\begin{theorem}\label{thm:proptrans}
	\label{thm:expresstrans}
    The membership problem of transducers to the classes of determinisable,
    functional, $k$-sequential, multi-sequential, $k$-valued,
    and finite-valued transducers can be reduced to the
    model-checking problem of $\logic_\trans$ with a constant space
    reduction. The obtained formulas are constant (as
    long as $k$ is fixed). 
\end{theorem}

\begin{proof}
    Without going through all the properties,
    let us remind the reader that the formula for
    $k$-valuedness has been given in the introduction. We also 
    give the $\logic_\trans$ formulas for
    the class of  determinisable transducer. It is known that a
    transducer is determinisable iff it satisfies the twinning
    property, which is literally the negation of:
    $$
    \begin{array}{lr}
      \qquad \exists \pi_1:q_1\xrightarrow{u\mid v_1} p_1,\exists \pi'_1:
      p_1\xrightarrow{u'\mid v'_1} p_1,      \exists
                                           \pi''_1:p_1\xrightarrow{u''\mid
                                           v''_1} r_1,\\
\qquad \exists \pi_2:q_2\xrightarrow{u\mid v_2} p_2,\exists \pi'_2:
      p_2\xrightarrow{u'\mid v'_2} p_2,\exists \pi''_2:
      p_2\xrightarrow{u''\mid v''_2} r_2, \smallskip\\ 
      \qquad\qquad\init(q_1)\wedge \init(q_2)\wedge \final(r_1)\wedge \final(r_2)\wedge\text{SDel}_{\neq}(v_1,v'_1,v_2,v'_2) & \qquad\qquad\hfill \square
    \end{array}
        $$
\end{proof}

\begin{theorem}\label{thm:maintrans}
	The model checking of transducers against formulas in $\logic_{Trans}$ is \textsc{PSpace-C}.
	It is in \textsc{NLogSpace-C} if the formula is fixed.
\end{theorem}

\begin{proof}[sketch]
    We use Parikh automata as acceptors for tuples of paths. They extend
    automata with counters that can only be incremented and never
    tested for zero. The acceptance condition is given by a
    semi-linear set (represented for instance by an existential
    Presburger formula). The formal definition can be found e.g. in~\cite{DBLP:conf/lics/FigueiraL15}. The counters allow us to compare the output length of
    paths, or to identify some output position of two paths with
    different labels (to test $v\not\prefix v'$). The counters are
    needed because this position may not occur at the same location in
    the convolution encoding of path tuples.$\hfill\square$
\end{proof}


\begin{corollary}[of Theorems~\ref{thm:expresstrans} and~\ref{thm:maintrans}]\label{cor:proptrans}
    The membership problem of transducers to the classes of determinisable,
    functional, $k$-sequential, multi-sequential, $k$-valued,
    and finite-valued transducers (for fixed $k$) is decidable in
    \textsc{NLogSpace}.
\end{corollary}

	\section{A pattern logic for sum-automata}\label{sec:sum}

We remind the reader that sum-automata are automata with outputs in
the monoid $\M_\suma=(\Z,+,0)$ (assumed to be encoded in binary) and therefore define subsets of $\Lambda^* \times \mathbb{Z}$.  We consider in this
section two logics for expressing structural properties of
sum-automata:  the logic $\logic_\suma$ which is obtained as
$\logic[\{\leq\}]$ where the output predicate $\leq$ is interpreted by
the natural total order over integers, and a subset of this logic
$\logic_\suma^{\neq}$ obtained as $\logic^+[\{\neq\}]$
where the predicate $\neq$ never
appears in the scope of an odd number of negations (to avoid the
expressibility of the equality predicate). We show that the fragment
$\logic_\suma^{\neq}$ enjoys better complexity results.
Formally, those two logics are defined as follows:


\begin{definition}[Two pattern logics for sum-automata]
	The logic $\logic_\suma$ is the set of formulas of the form
	$$
	\begin{array}{lllllll}
	\varphi & \Coloneqq &
	\exists \pi_1\colon p_1\xrightarrow{u_1\mid v_1} q_1,
	\dots,
	\exists \pi_n\colon p_n\xrightarrow{u_n\mid v_n} q_n,
	\C
	\smallskip\\
	\C & \Coloneqq &
	\lnot \C \mid
	\C \lor \C \mid
	u \prefix u' \mid
	u \in L \mid
	|u| \leq |u'|\mid
	\init(q) \mid
	\final(q) \mid
	q = q' \mid
	\pi=\pi' |
	\\
	&&
	t \leq t'
	\end{array}
	$$
	where for all $1\leq i<j\leq n$, $\pi_i\neq \pi_j$, $L$ is a regular
	language over $\Lambda$ assumed to be represented as an NFA,
	$u,u'\in\{u_1,\dots,u_n\}$, $q,q'\in\{q_1,\dots,q_n\}$,
	$t,t'\in Terms(\{v_1,\dots,v_n\},\cdot,\epsilon)$ and
	$\pi,\pi'\in\{\pi_1,\dots,\pi_n\}$.
	
	The logic $\logic^{\neq}_\suma$ is defined as above but the
        constraint $t \leq t'$ is replaced by $t \neq t'$ and this
        constraint does not occur under an odd number of negations,
        and moreover $v_i\neq v_j$ for all $1\leq i<j\leq n$ (no implicit output equality tests).  
\end{definition}

We review here some of the main sum-automata subclasses decidable in
{\sc PTime} studied in the
literature. We refer the reader to the mentioned references for the
formal definitions. The class of {\em functional} sum-automata~\cite{DBLP:journals/corr/abs-1111-0862}
are those such that all accepting paths associated with a given word return the
same value. The classes of {\em $k$-valued}~\cite{DBLP:conf/fsttcs/FiliotGR14} and
{\em $k$-sequential} sum-automata~\cite{DBLP:conf/fossacs/DaviaudJRV17} are defined similarly as for
transducers.  

\begin{theorem}\label{thm:propsum}
\label{thm:PLsum-expr}
The membership problem of sum-automata in the class of functional,
$k$-valued, and $k$-sequential automata can be reduced to the
model-checking problem of $\logic^{\neq}_\suma$. Moreover, the
obtained $\logic^{\neq}_\suma$ formulas are constant (as long as $k$ is fixed).
\end{theorem}
\begin{proof}
We have already shown in the introduction that functionality~\cite{DBLP:journals/corr/abs-1111-0862} and more generally $k$-valuedness~\cite{DBLP:conf/fsttcs/FiliotGR14} are expressible in $\logic^{\neq}_\suma$. The twinning property~\cite{DBLP:journals/corr/abs-1111-0862,DBLP:journals/jalc/AllauzenM03} is as well expressible in
    $\logic^{\neq}_\suma$, just by replacing in the formula expressing
    it for transducers (proof of Thm.~\ref{thm:expresstrans}) the atom
    $\text{SDel}_{\neq}(v_1,v'_1,v_2,v'_2)$ by $v'_1\neq v'_2$. 
In~\cite{DBLP:conf/fossacs/DaviaudJRV17}, a generalization of the
twinning property is shown to be complete for testing
$k$-sequentiality.$\hfill\square$
\end{proof}

The proof of the results below for $\logic_\suma$ follows arguments
that are similar to those developed for transducers in the proof of
Theorem~\ref{thm:maintrans}, and for the {\sc PTime} result for
$\logic^{\neq}_\suma$, we use a reduction to the $k$-valuedness
problem of sum-automata \cite{DBLP:conf/fsttcs/FiliotGR14}.

\begin{theorem}
\label{thm:PLsum-complexity}
	The model checking of sum-automata against formulas in $\logic_\suma$ is \textsc{PSpace-C}, \textsc{NP-C} when the formula is fixed, and {\sc NLogSpace-C} if in addition the values of the automaton are encoded in unary. The model checking of sum-automata against formulas in $\logic^{\neq}_\suma$ is \textsc{PSpace-C}, and in \textsc{PTime} when the formula is fixed (even if the values of the automaton are encoded in binary). 
\end{theorem}


\begin{corollary}[of Theorems~\ref{thm:PLsum-expr} and~\ref{thm:PLsum-complexity}]\label{cor:propsum}
The membership problem of sum-automata in the class of functional, $k$-valued, and $k$-sequential automata is decidable in {\sc PTime}.
\end{corollary}

Note that we have shown that the $k$-valuedness property is
expressible in $\logic^{\neq}_\suma$, and so the $k$-valuedness
property is reducible to the model-checking problem of
$\logic^{\neq}_\suma$. Nevertheless, this result does not provide a
new algorithm for $k$-valuedness as our model-checking
algorithm is based on a reduction to
$k$-valuedness~\cite{DBLP:conf/fsttcs/FiliotGR14}.

%

%
%

    \vspace{-2mm}
\section{Extensions and Future Work}
\vspace{-2mm}

The logics we have presented can be extended in two ways by keeping
the same complexity results, no matter what the output monoid is. The first extension allows to express properties of automata
whose states can be coloured by an arbitrary (but fixed) set of
colours. This is useful for instance to express properties of disjoint
unions of automata, the colours allowing to identify the
subautomata. The second extension is adding a bunch of universal state
quantifiers before the formula. This does not change the complexity,
and allow for instance to express properties such as whether an
automaton is trim (all its states are accessible and
co-accessible). As future work, we would like to investigate other
monoids (discounted sum group for
instance~\cite{DBLP:journals/corr/abs-1111-0862}), and other data
structures for which transducers and weighted automata have been
defined: nested words, infinite words and trees are the main
structures we want to work on.

\vspace{-4mm}
 





	
\end{document}